\documentclass[12pt,twoside]{article}

\usepackage{amsmath,amsthm,amssymb}
\newtheorem{theorem}{Theorem}[section]
\newtheorem{corollary}[theorem]{Corollary}
\newtheorem{lemma}[theorem]{Lemma}

\newtheorem{remark}[theorem]{Remark}

\numberwithin{equation}{section}

\newcommand\R{{\mathbb R}}
\newcommand\X{{{\mathbb R}^d}}
\newcommand\N{{\mathbb N}}

\newcommand\F{{\mathcal F}}

\newcommand\M{{\mathcal M}}
\newcommand\B{{\mathcal B}}
\newcommand\Ga{\Gamma}
\newcommand\ga{\gamma}
\newcommand\La{\Lambda}
\newcommand\la{\lambda}

\newcommand\1{{1 \!\! 1}}

\newcommand\Ls{L_{\sigma}}
\newcommand\hLs{\hat{L}_{\sigma}}

\newcommand\hLsd{\hLs^*}
\newcommand\Lsm{L_{\sigma,\,m}}
\newcommand\hLsm{\hat{L}_{\sigma,\,m}}

\newcommand\hLsmd{\hLsm^*}

\allowdisplaybreaks[3]

\title{Regulation mechanisms in spatial stochastic development models}

\author{Dmitri Finkelshtein\thanks{Institute of Mathematics,
National Academy of Sciences of Ukraine, 01601, Kyiv, Ukraine ({\tt
fdl@imath.kiev.ua}).} \and Yuri Kondratiev\thanks{Fakult\"{a}t
f\"{u}r Mathematik, Universit\"{a}t Bielefeld, 33615 Bielefeld,
Germany ({\tt kondrat@mathematik.uni-bielefeld.de}).}}

\date{}

\begin{document}

\maketitle

\begin{abstract}
The aim of this paper is to analyze different regulation mechanisms
in spatial continuous stochastic development models.  We describe
the density behavior  for  models with  global mortality  and local
establishment rates. We prove that the local self-regulation via a
competition mechanism (density dependent mortality) may suppress a
unbounded growth of the averaged density if the competition kernel
is superstable.
\end{abstract}



\pagestyle{myheadings} \thispagestyle{plain}
\markboth{D.~Finkelshtein, Yu.~Kondratiev}{Regulation mechanisms in
spatial stochastic development models}

\section{Introduction}

We will discuss some classes on interacting particle systems (IPS)
located in the Euclidean space $\X$. The phase space of such system
is the configuration space $\Gamma=\Gamma(\X)$ on $\X$. By
definition, each configuration $\gamma\in \Gamma$ is a locally
finite subset $\gamma\subset \X$. So, due to the standard
terminology, we will deal with continuous systems. Random evolutions
of IPS are given by Markov processes on $\Gamma$. Between all such
processes, one may distinguish a subclass of so-called spatial
birth-and-death Markov processes. In these processes points of a
randomly evolving configuration appear and disappear due to a Markov
rule (see (3.1) below).  Particular types of spatial birth-and-death
processes are motivated by several applications. For example,
Glauber type dynamics for classical continuous gases belongs to this
class \cite{BCC}, \cite{KoLy}. Another very essential source of such
processes is given by individual based models in spatial ecology or
agent based models in socio-economic systems, see, e.g.,
\cite{FKoKu} and the references therein. In any case, a concrete
form of birth and death rates in the stochastic dynamics should
reflect a microscopic structure of the system under consideration.

To describe the problem we are going to analyze, let us start with
the simplest case of a pure birth stochastic Markov process. In this
process, new points appear in the configuration independently of
existing points and locations of these new points are uniformly
distributed in the space. A possible interpretation of such random
evolution is related to an independent creation of identical
economic units in the space without any influence of their spatial
locations. We will call this process the free stochastic development
model. Another motivation to study such stochastic evolutions comes
from applications of the free development dynamics to generalized
mutation-selection models in mathematical genetics. In these models
a configuration describes locations of mutations inside of a genom
and new mutations spontaneously appear and are equally distributed
in the genom, see \cite{SEW}, \cite{KoMiPi}, \cite{KoKuOh}.
Obviously, this process is monotonically growing  and it is easy to
show that the density of particles in such a system will linearly
grows in time. We would like to answer the following question: How
may global regulations and local interactions change the asymptotic
behavior of the system? More precisely, we will consider three
particular cases of stochastic development models including:

(i) A global regulation via a mortality rate that prescribes to
particles (economic units)  random life times (exponentially i.i.d.
with a parameter $m>0$). This case corresponds to the well-known
Surgailis independent birth-and-death Markov processes on $\Gamma$,
see, e.g., \cite{Sur83}, \cite{Sur84}, \cite{KLR08}. In the
framework of mathematical physics, it is just Glauber dynamics for
classical free gas.

(ii) An establishment effect. In this case, the distribution of the
position of a new particle depends on the local structure of the
configuration. Newborn units will appear with small intensity in
densely occupied regions. We will see that the establishment itself
is not enough to prevent the growth of the density in the system. In
fact, the establishment effects  slower (logarithmic) growth
contrary to the linear growth in the free model.

(iii) A self-regulation via competition. The competition is
reflected in a density dependent mortality.  The latter means that
the mortality rate for each unit depends on the local structure of
the configuration around this unit. The mortality rate enters into
the model as a relative energy of a particle inside the
configuration corresponding  to a competition potential. The
described competition mechanism provides only a local regulation of
dense regions inside a configuration. Nevertheless, Theorem 4.2
shows a global bound for the averaged density of the stochastic
development process. Note that for the proof of this result we use
an assumption of positive definiteness (and, as a consequence,
superstability) of the competition potential in the form stated in
\cite{LPdS84}. Therefore, the main result concerning the competition
case may be read as follows: a properly organized competition in the
stochastic development systems produces a self-regulation for the
density of units.

Let us stress an essential point concerning the main aim of this
paper. At present, we have quite restrictive conditions for the
existence of general spatial birth-and-death processes, see e.g.
\cite{GarKur06}. In many  applications we need a weaker information.
Namely, we are interesting in  the existence of  Markov functions
corresponding to given birth and death rates and a certain class of
initial distributions on $\Gamma$. These Markov functions and their
one-dimensional distributions are enough to describe the time
evolution of initial states of systems and to analyze asymptotic
properties of stochastic dynamics (invariant states, ergodicity
etc.).  In the case of infinite particle systems, the concept of
Markov functions is strictly weaker than that of Markov processes,
and for particular models considered below there exist constructive
methods which solve the existence problem, see \cite{FKoKu},
\cite{GarKur06}, \cite{KoKutZh06} and \cite{KoKuMi}. But the  main
aim of our analysis is to obtain an {\it a priori\/} information
about the time-space behavior of such important characteristics of
these processes as correlations functions of their one-dimensional
distributions which are probability measures on $\Gamma$. In
particular, we are interested in the behavior of the particle
density in course of the stochastic evolution. A constructive
possibility to obtain some {\it a priori\/} bounds on
characteristics of Markov dynamics may be also realized in other
interesting IPS. As an example, we can mention the Dieckmann--Law
model in spatial ecology where the existence problem remains open
but conditions for explosion and non-explosion are stated in terms
of the parameters of the model in \cite{FK-DL}.  Moreover,
analogously to the well-known situation in the PDE theory, {\it a
priori\/} bounds may play a crucial role in the study of the
existence problem.

\section{General facts and notations}

Let $\B({\R}^{d})$ be the family of all Borel sets in $\R^d$ and
$\B_{b}({\R}^{d})$ denotes the system of all bounded sets from
$\B({\R}^{d})$.

The space of $n$-point configuration is
\[
\Ga _{0}^{(n)}=\Ga _{0,{\R}^{d}}^{(n)}:=\left\{ \left. \eta \subset
{\R}^{d}\right| \,|\eta |=n\right\} ,\quad n\in \N_0:=\N\cup \{0\},
\]
where $|A|$ denotes the cardinality of the set $A$. The space
$\Ga_{\La}^{(n)}:=\Ga _{0,\La }^{(n)}$ for $\La \in \B_b({\R}^{d})$
is defined analogously to the space $\Ga_{0}^{(n)}$. As a set, $\Ga
_{0}^{(n)}$ is equivalent to the symmetrization of
\[
\widetilde{({\R}^{d})^n} = \left\{ \left. (x_1,\ldots ,x_n)\in
({\R}^{d})^n\right| \,x_k\neq x_l\,\,\mathrm{if} \,\,k\neq l\right\}
,
\]
i.e. $\widetilde{({\R}^{d})^n}/S_{n}$, where $S_{n}$ is the
permutation group over $\{1,\ldots,n\}$. Hence one can introduce the
corresponding topology and Borel $\sigma $-algebra, which we denote
by $O(\Ga_{0}^{(n)})$ and $\B(\Ga_{0}^{(n)})$, respectively. Also
one can define a measure $m^{(n)}$ as an image of the product of
Lebesgue measures $dm(x)=dx$ on $\bigl(\R^d, \B(\R^d)\bigr)$.

The space of finite configurations
\[
\Ga _{0}:=\bigsqcup_{n\in \N_0}\Ga _{0}^{(n)}
\]
is equipped with the topology which has structure of disjoint union.
Therefore, one can define the corresponding Borel $\sigma $-algebra
$\B (\Ga _0)$.

A set $B\in \B (\Ga _0)$ is called bounded if there exists $\La \in
\B_b({\R}^{d})$ and $N\in \N$ such that $B\subset
\bigsqcup_{n=0}^N\Ga _\La ^{(n)}$. The Lebesgue---Poisson measure
$\la_{z} $ on $\Ga_0$ is defined as
\[
\la _{z} :=\sum_{n=0}^\infty \frac {z^{n}}{n!}m ^{(n)}.
\]
Here $z>0$ is the so-called activity parameter. The restriction of
$\la _{z} $ to $\Ga _\La $ will be also denoted by $\la _{z} $. We
denote  $\lambda:=\lambda_1$.

The configuration space
\[
\Ga :=\left\{ \left. \ga \subset {\R}^{d}\ \right| \; |\ga \cap \La
|<\infty, \text{ for all } \La \in \B_b({\R}^{d})\right\}
\]
is equipped with the vague topology. It is a Polish space (see, e.g.,
\cite{KoKut}). The corresponding  Borel $\sigma $-algebra $ \B(\Ga
)$ is defined as the smallest $\sigma $-algebra for which all
mappings $N_\La :\Ga \rightarrow \N_0$, $N_\La (\ga ):=|\ga \cap \La
|$ are measurable, i.e.,
\[
\B(\Ga )=\sigma \left(N_\La \left| \La \in
\B_b({\R}^{d})\right.\right ).
\]
One can also show that $\Ga $ is the projective limit of the spaces
$\{\Ga _\La \}_{\La \in \B_b({\R}^{d})}$ w.r.t. the projections
$p_\La :\Ga \rightarrow \Ga _\La $, $p_\La (\ga ):=\ga _\La $, $\La
\in \B_b({\R}^{d})$.

The Poisson measure $\pi _{z} $ on $(\Ga ,\B(\Ga ))$ is given as the
projective limit of the family of measures $\{\pi _{z} ^\La \}_{\La
\in \B_b({\R}^{d})}$, where $\pi _{z} ^\La $ is the measure on $\Ga
_\La $ defined by $\pi _{z} ^\La :=e^{-z m (\La )}\la _{z}$.

We will use the following classes of functions:
$L_{\mathrm{ls}}^0(\Ga _0)$ is the set of all measurable functions
on $\Ga_0$ which have a local support, i.e. $G\in
L_{\mathrm{ls}}^0(\Ga _0)$ if there exists $\La \in \B_b({\R}^{d})$
such that $G\upharpoonright_{\Ga _0\setminus \Ga _\La }=0$;
$B_{\mathrm{bs}}(\Ga _0)$ is the set of bounded measurable functions
with bounded support, i.e. $G\upharpoonright_{\Ga _0\setminus B}=0$
for some bounded $B\in \B(\Ga_0)$.

On $\Ga $ we consider the set of cylinder functions
$\mathcal{F}L^0(\Ga )$, i.e. the set of all measurable functions $G$
on $\bigl(\Ga,\B(\Ga))\bigr)$ which are measurable w.r.t. $\B_\La
(\Ga )$ for some $\La \in \B_b({\R}^{d})$. These functions are
characterized by the following relation: $F(\ga )=F\upharpoonright
_{\Ga _\La }(\ga _\La )$.

The following mapping between functions on $\Ga _0$, e.g.
$L_{\mathrm{ls}}^0(\Ga _0)$, and functions on $\Ga $, e.g.
$\mathcal{F}L^{0}(\Ga )$, plays the key role in our further
considerations:
\begin{equation}
KG(\ga ):=\sum_{\eta \Subset \ga }G(\eta ), \quad \ga \in \Ga,
\label{KT3.15}
\end{equation}
where $G\in L_{\mathrm{ls}}^0(\Ga _0)$, see e.g.
\cite{KoKu99,Le75I,Le75II}. The summation in the latter expression
is taken over all finite subconfigurations of $\ga ,$ which is
denoted by the symbol $\eta \Subset \ga $. The mapping $K$ is
linear, positivity preserving, and invertible, with
\begin{equation}
K^{-1}F(\eta ):=\sum_{\xi \subset \eta }(-1)^{|\eta \setminus \xi
|}F(\xi ),\quad \eta \in \Ga _0.\label{k-1trans}
\end{equation}

Let $ \mathcal{M}_{\mathrm{fm}}^1(\Ga )$ be the set of all
probability measures $\mu $ on $\bigl( \Ga, \B(\Ga) \bigr)$ which
have finite local moments of all orders, i.e. $\int_\Ga |\ga _\La
|^n\mu (d\ga )<+\infty $ for all $\La \in \B_b(\R^{d})$ and $n\in
\N_0$. A measure $\rho $ on $\bigl( \Ga_0, \B(\Ga_0) \bigr)$ is
called locally finite iff $\rho (A)<\infty $ for all bounded sets
$A$ from $\B(\Ga _0)$. The set of such measures is denoted by
$\mathcal{M}_{\mathrm{lf}}(\Ga _0)$.

One can define a transform $K^{*}:\mathcal{M}_{\mathrm{fm}}^1(\Ga
)\rightarrow \mathcal{M}_{\mathrm{lf}}(\Ga _0),$ which is dual to
the $K$-transform, i.e., for every $\mu \in
\mathcal{M}_{\mathrm{fm}}^1(\Ga )$, $G\in \B_{\mathrm{bs}}(\Ga _0)$
we have
\[
\int_\Ga KG(\ga )\mu (d\ga )=\int_{\Ga _0}G(\eta )\,(K^{*}\mu
)(d\eta ).
\]
The measure $\rho _\mu :=K^{*}\mu $ is called the correlation
measure of $\mu $.

As shown in \cite{KoKu99} for $\mu \in
\mathcal{M}_{\mathrm{fm}}^1(\Ga )$ and any $G\in L^1(\Ga _0,\rho
_\mu )$ the series \eqref{KT3.15} is $\mu $-a.s. absolutely
convergent. Furthermore, $KG\in L^1(\Ga ,\mu )$ and
\begin{equation}
\int_{\Ga _0}G(\eta )\,\rho _\mu (d\eta )=\int_\Ga (KG)(\ga )\,\mu
(d\ga ). \label{Ktransform}
\end{equation}

A measure $\mu \in \mathcal{M}_{\mathrm{fm} }^1(\Ga )$ is called
locally absolutely continuous w.r.t. $\pi _{z} $ iff $\mu_\La :=\mu
\circ p_\La ^{-1}$ is absolutely continuous with respect to $\pi
_{z} ^\La $ for all $\La \in \B_\La ({\R}^{d})$. In this case $\rho
_\mu :=K^{*}\mu $ is absolutely continuous w.r.t $\la$. We define
correlation functional $k_\mu:\Ga_0\rightarrow \R_+$ corresponding
to the~measure $\mu$ as the~Radon---Nikodym derivative:
\[
k_{\mu}(\eta):=\frac{d\rho_{\mu}}{d\la}(\eta),\quad \eta\in\Ga_{0}.
\]
The correlation functional $k_\mu$ may be considered as the system
of correlation functions corresponding to the restrictions
$k_\mu\upharpoonright_{\Ga^{(n)}}$. These functions are defined as
following
\begin{equation}
k_{\mu}^{(n)}:(\R^{d})^{n}\longrightarrow\R_{+}
\end{equation}
\[ k_{\mu}^{(n)}(x_{1},\ldots,x_{n}):=\left\{\begin{array}{ll}
k_{\mu}(\{x_{1},\ldots,x_{n}\}), & \mbox{if $(x_{1},\ldots,x_{n})\in
\widetilde{(\R^{d})^{n}}$}\\ 0, & \mbox{otherwise}
\end{array}
\right.
\]
and they are well known in statistical physics, see e.g \cite{R69},
\cite{R70}. In applications a specially important role play
correlation functions of the first and second orders: $k^{(1)}(x)$
and $k^{(2)}(x, y)$. These functions describe, respectively, the
density of particles and pair correlations.

A measure $\mu \in \mathcal{M}_{\mathrm{fm} }^1(\Ga )$ is called
translation invariant if it is invariant with respect to shifts of
configurations $\Ga\ni\ga\mapsto\{x+a \,|\, x\in\ga\}\in\Ga$ for any
$a\in\R^d$. The first-order correlation function of such measure
doesn't depend on the space coordinate: $k^{(1)}(x)\equiv
k_\mu^{(1)}$; and the second-order correlation function depends on
difference of coordinates: $k_\mu^{(2)}(x,y)=k_\mu^{(2)}(x-y)$.

\section{Stochastic development models}
Spatial birth-and-death processes describe dynamics of
configurations in $\R^d$ when particles (agents, companies, economic units)
disappear (die) from configurations and, on the other hand, some new
particles appear (born) somewhere in the space. The generator of
spatial birth-and-death dynamics is heuristically given on
measurable functions $F:\Ga\rightarrow\R$ by
\begin{multline}\label{genBaD}
( LF) ( \ga )  =\sum_{x\in \ga } d(x,\ga\setminus x)\bigl[
F( \ga \setminus x) -F( \ga ) \bigr]  \\
+\int_{\R^{d}} b(x,\ga) \bigl[ F( \ga \cup x) -F( \ga ) \bigr] dx,
\end{multline}
where $d,b:\X\times\Ga\to[0,\,\infty]$ are measurable rates of  death
and birth respectively. Of course, these rates should be finite
for a.a. $\ga\in\Ga$ with respect to a proper measure. Suppose that,
additionally,  $b(\cdot,\gamma) \in L^1_{loc}(\X)$  for a.a.
$\ga\in\Ga$. Then this  operator is well-defined at least on $\F
L^0(\Ga)$. Indeed, for $F\in\F L^0(\Ga)$
\[
F( \ga \setminus x ) - F ( \ga )= F( \ga \cup x ) - F ( \ga )=0,
\quad x\in\La^c:=\R^d\setminus\La,
\]
and the both terms of~\eqref{genBaD} are finite.


Note that if $L$ is a generator of such process (even if we know
that this process exists) then for the study of properties of the
corresponding stochastic dynamics we need some information about the semigroup
associated with $L$. This semigroup determines a solution to the
Kolmogorov equation which  has the following form:
\begin{equation}
\frac{dF_t}{dt}=L F_t,\qquad F_t\bigm|_{t=0}=F_0.\label{Kolmogor}
\end{equation}
In various applications the evolution of the corresponding
correlation functions (or measures) helps already to understand the
behavior of the process. The evolution of the correlation functions
of the process is related  to the evolution of states
of the system. The latter evolution is  given as a solution
to the dual Kolmogorov equation:
\begin{equation}
\frac{d\mu_t}{dt}=L^* \mu_t, \qquad
\mu_t\bigm|_{t=0}=\mu_0,\label{FokkerPlanc}
\end{equation}
where $L^*$ is the adjoint operator to $L$ on
$\M^1_\mathrm{fm}(\Ga)$, provided, of course, that it exists.

Using explicit form of $\hat{L}$ we derive the evolution equation
for \textit{quasi-observables} (functions on $\Ga_{0}$)
corresponding to the Kolmogorov equation \eqref{Kolmogor}. It has
the following form
\begin{equation}
\frac{dG_t}{dt}=\widehat{L}G_t,\qquad
G_t\bigm|_{t=0}=G_0.\label{quasiKolmogor}
\end{equation}
Then in the way analogous to those in which  the equation \eqref{FokkerPlanc} was determined for
\eqref{Kolmogor}, we get an evolution equation for the correlation
functions corresponding to the equation \eqref{quasiKolmogor}:
\begin{equation}
\frac{dk_t}{dt}=\widehat{L}^{*}k_t,\qquad
k_t\bigm|_{t=0}=k_0.\label{corrfunctiona}
\end{equation}
The generator $\hat{L}^*$ here is the dual to $\hat{L}$ w.r.t. the duality
 given by the following expression:
\begin{equation}
\left\langle\!\left\langle G,\,k\right\rangle\!\right\rangle =
\int_{\Ga_{0}}G\cdot k\, d\la \label{duality}.
\end{equation}

\paragraph{Free development model}

A simplest model in considered framework  is a model of free
development when particles (identical economic units) are born
independently without any influence of existing ones.
An interpretation is that a "decision" about appearing of a new
company is produced outside of the system and it is not motivated by
the situation inside of the system. Moreover, in this simplest model
particles (units) will not die (will not become bankrupts).

The formal pre-generator of the Markov  dynamics that describes such model is
the following:
\begin{equation*}
\left( \Ls F\right) \left( \gamma \right) =\sigma \int_{\mathbb{R}%
^{d}}\left[ F\left( \gamma \cup x\right) -F\left( \gamma \right)
\right] dx,
\end{equation*}%
where $\sigma >0$ is the intensity rate of new units creation.

It is easy to see that the operator $\Ls $ is well defined, for
example, on the set $\F L^0(\Ga)$. The corresponding Markov  process
exists due to, e.g., \cite{GarKur06}. Using results from
\cite{FiKoOl07}, we obtain
\begin{equation}
\bigl( \hLs G\bigr) \left( \eta \right) =\sigma \int_{%
\mathbb{R}^{d}}G\left( \eta \cup x\right) dx.\label{exprL+}
\end{equation}%
and
\begin{equation} \bigl( \hLsd k\bigr) \left( \eta \right) =\sigma
\sum_{x\in \eta }k\left( \eta \setminus x\right) .\label{exprL+dual}
\end{equation}

Immediately from \eqref{exprL+dual} and \eqref{corrfunctiona} (see
also \cite{FK08c}) we conclude that the density of the free
development model has the form
\[
k_t^{(1)}(x) = k_0^{(1)}(x) + \sigma t.
\]
Therefore, the density has linear growth in time. To prevent this
growth we need to modify the generator introducing some regulation
mechanisms in the model.

\paragraph{Development model with global regulation}  Below we
consider a model with a global regulation reflected in the death rate by an assumption
about a finite life time for economic units. More precisely, we assume that each point of the
configuration has exponentially distributed (with some positive
parameter $m$) random life time and these random times are independent. Hence,  a death (bankruptcy)
appears  as a random event equally distributed for all economic units independently of their
space locations.

A pre-generator describing such process has the following form:
\begin{equation*}
\left( \Lsm F\right) \left( \gamma \right) =m\sum_{x\in\ga}\bigl[F(\ga\setminus x
) - F(\ga) \bigr] + \sigma \int_{\mathbb{R}%
^{d}}\left[ F\left( \gamma \cup x\right) -F\left( \gamma \right)
\right] dx.
\end{equation*}
Note that the expression for $\Lsm$ coincides with the one for the generator of so-called
Surgailis process (see \cite{Sur83}, \cite{Sur84}, \cite{KLR08}).
Again, using results from \cite{FiKoOl07}, we obtain
\begin{align}
\bigl( \hLsm G\bigr) \left( \eta \right) &=-m|\eta| G(\eta)+\sigma \int_{%
\mathbb{R}^{d}}G\left( \eta \cup x\right) dx,\label{exprLP} \\
\bigl( \hLsmd k\bigr) \left( \eta \right) &=-m|\eta| k(\eta)+\sigma
\sum_{x\in \eta }k\left( \eta \setminus x\right) .\label{exprLPdual}
\end{align}

The considered stochastic dynamics has a unique invariant measure which is
the Poisson measure on $\Gamma$ with constant intensity $\frac{\sigma}{m}$.

Using \eqref{exprLPdual} (see also \cite{FK08c}) one can obtain a precise expression for
the density of the process:
\[
k^{(1)}_t(y)=e^{-tm}k^{(1)}_{0}\left(  y \right) +\frac{\sigma }{m} \left(
1-e^{-tm}\right).
\]
Therefore, for an initial state with bounded density  any positive global regulation rate $m$ gives time-space
bounded density  which converges uniformly in space to the limiting Poisson density. (For properties
of   higher order correlation functions see \cite{FK08c}.)

\paragraph{Establishment effects in the development model}
As we pointed out before, in the free development model an appearing of a new unit on the market and its location are
independent of the presented configuration of the system. A reasonable generalization of this model is such that the newborn unit prefers to choose a location with smaller density of already existing units. The latter may be considered as higher probability to survive in less occupied regions. The corresponding term which decreases the birth rate of the
 generator in densely populated areas is called the establishment term. We
consider the special case when this rate has  exponential form, but our
considerations may be extended to more general establishment rates.

Let $0\leq \phi \in L^{1}\left( \mathbb{R}^{d}\right) $,
$\phi(-x)=\phi(x)$, $x\in\X$,  and
\[
\left( L_\phi F\right) \left( \gamma \right) =\int_{\mathbb{R}^{d}}
\exp\Bigl\{-\sum_{y\in\ga}\phi (x-y)\Bigr\} \left[ F\left( \gamma
\cup x\right) -F\left( \gamma \right) \right] dx.
\]%
We suppose that there exist the dynamics of measures $\mu_t$ and let
$k_t$ will be corresponding correlation functions. Actually, the
existence of a Markov process for considered case may be obtained
from \cite{GarKur06}. Using for any $\varphi\in C_0(\X)$ the
equality
\[
\frac{\partial }{\partial t}\int_\X k_{t}^{\left( 1\right) }\left(
x\right)\varphi(x)\,dx =\frac{\partial }{\partial t}\int_\Ga
\left\langle \varphi, \gamma \right\rangle d\mu_t(\ga)= \int_\Ga
L_\phi \left\langle \varphi, \gamma \right\rangle d\mu_t(\ga)
\]
we obtain, by Jensen's inequality,
\begin{multline*}
\frac{\partial }{\partial t}k_{t}^{\left( 1\right) }\left( x\right)
=\int_{\Gamma } \exp\Bigl\{-\sum_{y\in\ga}\phi (x-y)\Bigr\} d\mu
_{t}\left( \gamma \right) \\ \geq \exp\left( - \int_{\Gamma
}\sum_{y\in\ga}\phi (x-y) d\mu _{t}\left( \gamma \right) \right)
=\exp\left( -\int_{ \mathbb{R}^{d}}\phi \left( x-y\right)
k_{t}^{\left( 1\right) }\left( y\right) dy\right) .
\end{multline*}
In the translation invariant case we obtain
\[
\frac{d}{dt}k_{t}^{\left( 1\right) }\geq \exp\left( -\left\langle
\phi \right\rangle k_{t}^{\left( 1\right) }\right) ,
\]
where $\left\langle \phi \right\rangle = \int_\X \phi(x) \,dx$.
Hence, if $g_{t}$ is a positive solution of the equation
\[
\frac{d}{dt}g_{t}=e^{-\left\langle \phi \right\rangle g_{t}},
\]%
then $k_{0}^{\left( 1\right) }\geq g_{0}$ implies $k_{t}^{\left(
1\right) }\geq g_{t}$.

One has
\begin{align*}
e^{\left\langle \phi \right\rangle g_{t}}dg_{t} &=dt \\
\frac{1}{\left\langle \phi \right\rangle }e^{\left\langle \phi
\right\rangle
g_{t}} &=t+\frac{C}{\left\langle \phi \right\rangle } \\
g_{t} &=\frac{1}{\left\langle \phi \right\rangle }\ln \left(
\left\langle
\phi \right\rangle t+C\right) ,~~C>1 \\
g_{0} &=\frac{1}{\left\langle \phi \right\rangle }\ln C.
\end{align*}
Putting for any $k_{0}^{\left( 1\right) }$ the initial
value $g_{0} =k_{0}^{\left( 1\right) }$  we obtain that
\[
k_{t}^{\left( 1\right) }\geq \frac{1}{\left\langle \phi
\right\rangle }\ln \left( \left\langle \phi \right\rangle t+\exp
\left\{ k_{0}^{\left( 1\right) }\left\langle \phi \right\rangle
\right\} \right).
\]

Therefore, the establishment term cannot prevent unboundedness of
density. We may expect only essentially slower growth due to the
establishment effect.

\begin{remark}
Of course, if we consider two regulation mechanisms, namely, the global
regulation and the establishment together, then the first-order
correlation function will be also bounded (more precisely,
all correlation functions will have sub-Poissonian bounds, cf.
\cite{FK08c}). Moreover, in this case the operator
\begin{multline*}
\left( L_G F\right) \left( \gamma \right)
=m\sum_{x\in\ga}\bigl[F(\ga\setminus x ) - F(\ga) \bigr] \\ + \sigma
\int_{\mathbb{R} ^{d}}\left[ F\left( \gamma \cup x\right) -F\left(
\gamma \right) \right] \exp\Bigl\{-\sum_{y\in\ga}\phi (x-y)\Bigr\}
dx
\end{multline*}
is the generator of so-called Glauber dynamics for a classical gas model (see, e.g.,
\cite{KL05}, \cite{KoKutZh06}). If $\phi$ has some additional
properties such that there exists Gibbs measure with this potential,
then such measure will be invariant (and even symmetrizing one) for the generator $L_G$. On the
other hand, known properties of the corresponding Markov dynamics imply that
 corresponding correlation functions satisfied so-called
generalized Ruelle bounds (see \cite{KoKutZh06} for details).
\end{remark}

\section{Stochastic  development models with competitions }

In the previous section we considered, in particular, global
(outward) regulation in the model. As wee see, such regulation may
prevent unbounded (linear) growth (in time) of the density of our
system. In this section we consider the case of a local regulation
which appear due to the competition between elements (units) of
the system.  A
pre-generator which describes such model has the following form:
\begin{align*}
\left( L_{a,\sigma}F\right) \left( \gamma \right) &=\sum_{x\in
\gamma }\biggl( \sum_{y\in \gamma \setminus x}a\left( x-y\right)
\biggr) \left[ F\left(
\gamma \setminus x\right) -F\left( \gamma \right) \right] \\
&\quad+\sigma \int_{\R^d}\left[ F\left( \gamma \cup x\right) -F\left( \gamma
\right) \right] dx.
\end{align*}
Here $0\leq a\in L^1(\R^d)$ is an even function s.t.
\[
\left\langle a\right\rangle :=\int_{\mathbb{R}^{d}}a\left( x\right)
dx>0.
\]

The question about existence of a process with the generator
$L_{a,\sigma}$ we will not discuss in this paper. We just assume that there exist the
dynamics of measures $\mu_t$ and let $k_t$ will be the corresponding
correlation functional.

Using results from \cite{FiKoOl07} we obtain that
\begin{align*}
\bigl( \hat{L}_{a,\sigma}G\bigr) \left( \eta \right) &=-2E_{a}\left( \eta \right)
G\left( \eta \right) -\sum_{x\in \eta }\biggl( \sum_{y\in \eta \setminus
x}a\left( x-y\right) \biggr) G\left( \eta \setminus x\right) \\
&\quad +\sigma  \int_{\R^d}G\left( \eta \cup x\right) dx
\end{align*}%
and
\begin{align}
\bigl( \hat{L}_{a,\sigma}^*k\bigr) \left( \eta \right) &=-2E_{a}\left( \eta
\right) k\left( \eta \right) -\int_{\R^d}\sum_{y\in \eta }a\left( x-y\right)
k\left( \eta \cup x\right) dx \label{LdualSoc}\\
&\quad +\sigma  \sum_{x\in \gamma }k\left( \eta \setminus x\right)
,\nonumber
\end{align}%
where we used the following notations for the energy functional
corresponding to the pair potential $a(\cdot)$:
\[
E_a(\eta)=\sum_{\{x,y\}\subset\eta}a(x-y),
\quad \eta\in\Ga_0.
\]

It is easy to see that the Cauchy problems \eqref{quasiKolmogor} and
\eqref{corrfunctiona} for quasi-observables and correlation
functions respectively have a form of hierarchical chains and,
therefore, can not be solved explicitly. The latter is a common
problem in the study of stochastic dynamics of IPS. In several
particular models such as Glauber type dynamics in continuum
\cite{KoKuMi}, \cite{KoKutZh06} or some spatial ecological models
\cite{FKoKu} this difficulty may be overcame via a proper
perturbation theory techniques. As a result, in the mentioned works
we have existence results for corresponding evolutional equations
together with certain a-priori bound for the solutions. Note that
the perturbation techniques needs, in any case, a presence in the
system a small parameter. In the considered model such parameter is
clearly absent. Nevertheless, one can try to find estimate for the
correlation functions. Actually, in the presented below approach we
will use the explicit form of the Markov generator to obtain an
a-priori bound on the density of the system.


We will say that a sequence $\{\La_k,\, k\in\N\}$
of open bounded subsets of $\R^d$ is of~{\em F-type} if~$\bigcup_{k\in\N}\La_k=\R^d$,
$\La_k\subset\La_{k+1}$, $k\in\N$ and there
exists $F>0$ such that for any $h\in(0;1)$
and for any $k\in\N$
\[
s(\La_k,h):=\frac{|\La_k(h)\setminus\La_k|}{|\La_k|}
\leq F,
\]
where
\[
\Lambda_k (h):=\left\{ x:\inf_{y\in \Lambda_k }\left\vert x-y\right\vert <h\right\}.
\]

A simple example of {\em F-type}\/ sequence is
the sequence of balls $\La_k=B(0,k)$ with center at origin and radius $k\in\N$. Indeed,
for any $h<1$
\[
s \left( \Lambda_k ,h\right) =\frac{\left\vert
\Lambda_k(h)\setminus\Lambda_k\right\vert }{\left\vert \Lambda_{k}
\right\vert }-1=\frac{\left( R+h\right) ^{d}}{R^{d}}-1 =\left(
1+\frac{h}{R}\right) ^{d}-1<2^{d}-1.
\]

For any $\La\in\B_c(\R^d)$ we will call {\em
the average density} of the our system the
following object
\[
\rho_t^\La:=\frac{1}{|\La|}\int_\La k_t^{(1)}(x)dx,
\]
where $k_t^{(1)}$ is the first-order correlation
functions  (density) at moment $t\geq0$.

\begin{lemma}\label{lemSS}
Suppose that the function $a$ is continuous and positive definite
and the sequence $\{\La_k,k\in\N\}$ is F-type. Then there exists
$c>0$ such that for any open $\La\in\{\La_k,k\in\N\}$
\[
2E_a(\eta)\geq c\,\frac{|\eta|^2}{|\La|},
\quad \eta\in\Ga_\La.
\]
\end{lemma}
\begin{proof}
In \cite{LPdS84}, it was shown that for any continuous positive
definite function $a$ the energy $E_a$ is superstable, namely, for
any open $\La\subset\R^d$ and for any $\eta :=\left\{ x_{i}\right\}
_{i=1}^{n}\subset \Lambda $ the following inequality holds
\begin{equation*}
2E_{a}\left( \eta \right) \geq \frac{n^{2}}{\left\vert \Lambda \right\vert }%
\frac{\left[ \left\langle a\right\rangle -\delta \left( h\right) \right] ^{2}%
}{\left[ \left\langle a\right\rangle +\delta \left( h\right) +\sigma \left(
\Lambda ,h\right) \left\langle a\right\rangle \right] },
\end{equation*}%
where%
\begin{equation*}
\delta \left( h\right) =2\int_{\left\vert x\right\vert >h}a\left( x\right) dx\geq 0.
\end{equation*}

Therefore, for any $\La\in\{\La_k,k\in\N\}$
\[
2E_{a}\left( \eta \right) \geq \frac{n^{2}}{\left\vert \Lambda \right\vert
}\frac{\left[ \left\langle a\right\rangle -\delta \left( h\right) \right]
^{2}}{\bigl[ \delta \left( h\right) +(F +1)\left\langle a\right\rangle \bigr] } =:
\frac{n^{2}}{\left\vert \Lambda \right\vert
}c.
\]
Let $h\in(0;\,1)$ be such that
\[
 \left\langle a\right\rangle -\delta \left( h\right)= \int_{\left\vert x\right\vert \leq
h}a\left( x\right) dx -\int_{\left\vert x\right\vert >h}a\left( x\right) dx\neq0
\]
(we may always choose such $h$ since the first integral is an
increasing function of $h$ and the second one is a decreasing
function). Stress that $c>0$ and doesn't depend on $\La$.
\end{proof}
\begin{theorem}\label{thm-sm}
Suppose that the function $a$ is continuous and positive definite
and the sequence $\{\La_k,k\in\N\}$ is F-type; let $c$ be as in
Lemma~\ref{lemSS}. Suppose also that there exists
$D>\sqrt{\dfrac{\sigma}{c}}$ such that $\rho_0^{\La_k}\leq D$,
$k\in\N$. Then for any $t>0$, $k\in\N$
\[
\rho_t^{\La_k}\leq D.
\]
\end{theorem}
\begin{proof}
Note that for $F\left( \gamma \right) =\left\langle \varphi ,\gamma \right\rangle $, $\ga\in\Ga$,
$\varphi\in C_0(\R^d)$ we have
\begin{equation*}
\left( L_{a,\sigma}F\right) \left( \gamma \right) =-\sum_{x\in \gamma }\left(
\sum_{y\in \gamma \setminus x}a\left( x-y\right) \right) \varphi \left(
x\right) +\sigma  \int_{\R^d}\varphi \left( x\right) dx.
\end{equation*}%
Let $\varphi \left( x\right) =\1_{\Lambda }\left( x\right) $, $\Lambda \in
\{\La_k,k\in\N\} $. Then
$F(\ga)=|\ga_\La|$ and
\begin{align*}
\left( L_{a,\sigma}F\right) \left( \gamma \right) &=-\sum_{x\in \gamma }\left( \sum_{y\in \gamma \setminus x}a\left(
x-y\right) \right) \1_{\Lambda }\left( x\right) +\sigma  \left\vert \Lambda
\right\vert \\
&=-\sum_{x\in \gamma _{\Lambda }}\left( \sum_{y\in \gamma \setminus
x}a\left( x-y\right) \right) +\sigma  \left\vert \Lambda \right\vert \\
&\leq -\sum_{x\in \gamma _{\Lambda }}\left( \sum_{y\in \gamma _{\Lambda
}\setminus x}a\left( x-y\right) \right) +\sigma  \left\vert \Lambda
\right\vert \\
&=-2E_{a}\left( \gamma _{\Lambda }\right) +\sigma  \left\vert \Lambda
\right\vert \leq -\frac{c}{\left\vert \Lambda \right\vert }\left\vert \gamma _{\Lambda
}\right\vert ^{2}+\sigma  \left\vert \Lambda \right\vert.
\end{align*}

Let us set
\begin{align*}
n_{t}^{\Lambda }&:=\int_\Ga  \left\vert \ga_ \Lambda \right\vert d\mu_t (\ga)= \int_\Ga  \left\langle \1_\Lambda, \ga \right\rangle d\mu_t (\ga)\\& = \int_\X \1_\La(x) k_t^{(1)}(x)dx=\int_\La
k_t^{(1)}(x)dx = |\La|\rho_t^\La.
\end{align*}%
Then using Holder inequality
\begin{align*}
\frac{d}{dt}n_{t}^{\Lambda }
&=\int_{\Gamma} L_{a,\sigma}\left\vert \gamma _\Lambda
\right\vert d\mu _{t}\left( \gamma \right) \leq \int_{\Gamma}\left( \sigma  \left\vert \Lambda
\right\vert -\frac{c}{\left\vert \Lambda \right\vert }\left\vert \gamma
_{\Lambda }\right\vert ^{2}\right) d\mu _{t}\left( \gamma \right) \\
&=\sigma  \left\vert \Lambda \right\vert -\frac{c}{\left\vert \Lambda
\right\vert }\int_{\Gamma }\left\vert \gamma_ \Lambda
\right\vert ^{2}d\mu _{t}\left( \gamma \right) \leq \sigma  \left\vert \Lambda \right\vert -\frac{c}{\left\vert \Lambda
\right\vert }\left( \int_{\Gamma }\left\vert \gamma _\Lambda
\right\vert d\mu _{t}\left( \gamma \right) \right) ^{2} \\
&=\sigma  \left\vert \Lambda \right\vert -\frac{c}{\left\vert \Lambda
\right\vert }\left( n_{t}^{\Lambda }\right) ^{2}.
\end{align*}%
As a result,
\begin{equation*}
\frac{d}{dt}\rho _{t}^\La\leq \sigma  -c\bigl(\rho _{t}^\La\bigr)^{2}.
\end{equation*}%
Therefore, if we consider the positive solutions
of the Cauchy problem%
\begin{equation}\label{addCauchy}
\left\{
\begin{array}{l}
\dfrac{d}{dt}g\left( t\right) =\sigma  -cg^{2}\left( t\right) \\[2mm]
g\left( 0\right) =g_{0}
\end{array}%
\right.
\end{equation}%
with proper $g_0>0$ and if $\rho _{0}^\La\leq g_{0}$ then $\rho _{t}^\La\leq g\left( t\right) $, $t>0$. Solving
\eqref{addCauchy} we obtain
\begin{gather*}
\ln \frac{\left\vert \sqrt{c}g\left( t\right) +\sqrt{\sigma  }\right\vert }{%
\left\vert \sqrt{c}g\left( t\right) -\sqrt{\sigma  }\right\vert }-\ln \tilde{%
C}=2\sqrt{c\sigma  }t,~~~\tilde{C}>0; \\[2mm]
g\left( t\right) =\frac{Ce^{2\sqrt{c\sigma  }t}\sqrt{\sigma
}+\sqrt{\sigma
}}{Ce^{2\sqrt{c\sigma  }t}\sqrt{c}-\sqrt{c}}=\sqrt{\frac{\sigma  }{c }}%
\left( 1+\frac{2}{Ce^{2\sqrt{c\sigma  }t}-1}\right) ,~~C\in\R.
\end{gather*}%
Then%
\begin{equation*}
g\left( 0\right) =\sqrt{\frac{\sigma  }{c }}\left( 1+\frac{2}{C-1%
}\right) ,~~C\in\R.
\end{equation*}%
Let $g_{0}=D\geq\rho_0^\La$. Then since $D>\sqrt{\dfrac{\sigma}{c}}$
we have
\[
C=\frac{2}{D\sqrt{\dfrac{c}{\sigma}}-1}+1>1
\]
and
\[
Ce^{2\sqrt{c\sigma  }t}-1\geq C-1=\frac{2}{D\sqrt{\dfrac{c}{\sigma}}-1}>0.
\]
As a result,%
\begin{equation*}
\rho _{t}^\La\leq g(t) \leq \sqrt{\frac{\sigma  }{c }}\left( 1+\frac{2}{C-1 }%
\right) =D
\end{equation*}
for any $t>0$ and for any $\La\in\{\La_k,k\in\N\}$.
The statement is proved.
\end{proof}

\begin{corollary}\label{trinv1}
Under conditions of Theorem~\ref{thm-sm} in the translation
invariant case we have that $k_0^{(1)}\leq D$ implies $k_t^{(1)}\leq
D$.
\end{corollary}

At the end we consider a simple estimate for the second-order
correlation function. Let us suppose that
\begin{equation*}
a\left( u\right) >0, \quad u\in\X.
\end{equation*}%
Then in the translation invariant case the following estimate holds
\begin{align*}
k_{t}^{\left( 2\right) }\left( u\right) &\leq e^{-2a\left( u\right)
t}k_{0}^{\left( 2\right) }\left( u\right) +2\sigma  \int_{0}^{t}e^{-2a\left(
u\right) \left( t-\tau \right) }k_{\tau }^{\left( 1\right) }d\tau \\
&\leq e^{-2a\left( u\right) t}k_{0}^{\left( 2\right) }\left(
u\right) +2\sigma  D \int_{0}^{t}e^{-2a\left( u\right)
\left( t-\tau \right) }d\tau \\
&=e^{-2a\left( u\right) t}k_{0}^{\left( 2\right) }\left( u\right) +\frac{%
\sigma  D }{a\left( u\right) }\left( 1-e^{-2a\left( u\right)
t}\right)
\end{align*}

We have two possible estimates%
\begin{align}
k_{t}^{\left( 2\right) }(x-y)\leq e^{-2a\left( x-y\right)
t}k_{0}^{\left(
2\right) }\left( x-y\right) +\frac{\sigma  D }{ a\left( x-y\right) } \label{estshortdist}\\
\intertext{and} k_{t}^{\left( 2\right) }(x-y)\leq e^{-2a\left(
x-y\right) t}k_{0}^{\left( 2\right) }\left( x-y\right) +C\sigma  D
t.\label{estlongdist}
\end{align}

\textbf{Acknowledgments}. The financial support of DFG through the
SFB 701 (Bielefeld University) and German-Ukrainian Project 436 UKR
113/94 is gratefully acknowledged. This work was partially supported
by FCT, POCI2010, FEDER.


\addcontentsline{toc}{section}{References}

\begin{thebibliography}{GKKS01}

\bibitem{BCC} {\sc L.~Bertini, N.~Cancrini, F.~Cesi}, {\em The spectral
gap for a Glauber-type dynamics in a continuous gas}, {Ann. de
l'inst. H. Poincar\'{e} (B) Prob. et Stat.}, 38, no.~1 (2002),
pp.~91--108.

\bibitem{FK08c}
{\sc D.~L.~Finkelshtein, Yu.~G.~Kondratiev}, {\em Non-equilibrium
dynamics of~the~economic development model}, {In preparation}.

\bibitem{FK-DL}
{\sc D.~L.~Finkelshtein, Yu.~G.~Kondratiev}, {\em Dynamical
self-regulation in spatial population models in continuum}, {In
preparation}.

\bibitem{FKoKu} {\sc D.~L.~Finkelshtein, Yu.~G.~Kondratiev, O.~Kutovyi},
{\em Individual based model with competition in~spatial~ecology},
http://arxiv.org/abs/0803.3565. Submitted to: \textit{SIAM Journal
on Mathematical Analysis}.

\bibitem{FiKoOl07}
{\sc D.~L.~Finkelshtein, Yu.~G.~Kondratiev, and M.~J.~Oliveira},
{\em Markov evolutions and hierarchical equations in the continuum
I. One-component systems}, http://arxiv.org/abs/0707.0619. Submitted
to: \textit{Journal of Evolution Equations}.

\bibitem{GarKur06}
{\sc N.~L.~Garcia, T.~G.~Kurtz}, {\em Spatial birth and death
processes as solutions of stochastic equations}, Alea 1, (2006),
pp.~281--303.

\bibitem{KoKu99}
{\sc Yu.~G. Kondratiev and T.~Kuna}, {\em Harmonic analysis on
configuration space. I. General theory}, {Infinite Dimensional
Analysis, Quantum Probability and Related Topics}, {5, no.~2}
(2002), pp.~201-233.

\bibitem{KoKuOh}
{\sc Yu.~G. Kondratiev, T.~Kuna, and N.~Ohlerich}, {\em
Selection-mutation balance models with epistatic selection}, to
apper in: \textit{Condensed Matter Physics}.

\bibitem{KoKut}
{\sc Yu.~G.~Kondratiev and O.~V.~Kutoviy}, {\em On the metrical
properties of the configuration space}, {Math. Nachr.}, {279, no.~7}
(2006), pp.~774-783.

\bibitem{KoKutZh06}
{\sc Yu. Kondratiev, O.~Kutoviy, and E.~Zhizhina}, {\em
Nonequilibrium {G}lauber-type dynamics in continuum}, {J. Math.
Phys.}, 47(11):113501, 2006.

\bibitem{KL05}
{\sc {Yu}. Kondratiev and E.~Lytvynov}, {\em Glauber dynamics of
continuous particle systems}, {Ann. Inst. H. Poincar{\'e} Probab.
Statist.}, 41 (2005), pp.~685--702.

\bibitem{KoKuMi}
{\sc Yu.~G.~Kondratiev, O.~V.~Kutoviy, and R.~A.~Minlos}, {\em On
non-equilibrium stochastic dynamics for interacting particle systems
in continuum}, {Journal of Functional Analysis}, 255 (1), 2008,
pp.~200--227.

\bibitem{KoLy}
{\sc Yu.~G.~Kondratiev and E.~Lytvynov}, {\em Glauber dynamics of
continuous particle systems}, Ann. Inst. H.Poincare, Ser. A, Probab.
Statist. 41 (2005), pp.~685--702.

\bibitem{KLR08}
{\sc Yu.~G.~Kondratiev, E.~Lytvynov, M.~R\"{o}ckner}, {\em
Non-equilibrium stochastic dynamics in continuum: The free case.} To
appear In: {\it Proceedings of the International Conference on
Infinite Particle Systems, 8-11 October 2006, Kazimierz Dolny,
Poland}.

\bibitem{KoMiZh}
{\sc Yu.~G.~Kondratiev, R.~Minlos, and E.~Zhizhina}, {\em
One-particle subspaces of the generator of Glauber dynamics of
continuous particle systems}, Rev. Math. Phys., 16, no.~9 (2004),
pp.~1--42.

\bibitem{KoMiPi}
{\sc Yu.~G.~Kondratiev, R.~A.~Minlos, and S.~Pirogov}, {\em
Generalized mutation-selection model with epistatics}, Preprint,
2007.

\bibitem{Le75I}
{\sc A.~Lenard}, {\em States of classical statistical mechanical
systems of infinitely many particles. I},  Arch. Rational Mech.
Anal., 59 (1975), pp.~219-239.

\bibitem{Le75II}
{\sc A.~Lenard}, {\em States of classical statistical mechanical
systems of infinitely many particles. II}, Arch. Rational Mech.
Anal., 59 (1975), pp.~241-256.

\bibitem{LPdS84}
{\sc J. T. Lewis, J. V. Pul\`e, and P. de Smedt}, {\em The
superstability of pair-potential of positive type}, J. of Stat.
Physics, nos. 3/4, 35 (1984), pp.~381--385.

\bibitem{R69}
{\sc D.~Ruelle},
\newblock {\bf Statistical Mechanics}
\newblock (New York, Benjamin, 1969).

\bibitem{R70}
{\sc D.~Ruelle},
\newblock {\em Superstable interactions in classical statistical
mechanics},
\newblock {\em Commun.~Math.~Phys.}, 18 (1970) 127--159.

\bibitem{SEW} {\sc D.~Steinsaltz, S.~N.~Evans, and K.~W.~Wachrer},
{\em A generalized model of mutation-selection balance with
applications to aging}, {Adv. Appl. Math.}, 35, no. 1 (2005),
pp.~16--33.

\bibitem{Sur83}
{\sc D.~Surgailis}, \newblock {\em On Poisson multiple stochastic
integrals and associated equilibrium Markov processes}. In: Theory
and application of random fields (Bangalore, 1982), pp. 233--248,
Lecture Notes in Control and Inform. Sci., Vol. 49, Springer, Berlin
(1983).

\bibitem{Sur84}
{\sc D.~Surgailis}, \newblock {\em On multiple Poisson stochastic
integrals and associated Markov semigroups}. \newblock Probability
and mathematical statistics. Vol. 3, Fasc. 2 (1984), pp. 217--239.

\end{thebibliography}
\end{document}